\documentclass[conference]{IEEEtran}

\makeatletter
\def\ps@headings{%
\def\@oddhead{\mbox{}\scriptsize\rightmark \hfil \thepage}%
\def\@evenhead{\scriptsize\thepage \hfil \leftmark\mbox{}}%
\def\@oddfoot{}%
\def\@evenfoot{}}
\makeatother
\usepackage[dvips]{graphicx}
\usepackage{subfigure}
\usepackage{amsmath}
\usepackage{amsthm}
\usepackage{amssymb}
\usepackage{algorithmic, algorithm}

\newtheorem{theorem}{Theorem}

\newtheorem{lemma}{Lemma}

\renewcommand{\algorithmiccomment}[1]{//#1}
\renewcommand{\algorithmicrequire}{\textbf{Input:}}
\renewcommand{\algorithmicensure}{\textbf{Output:}}

\newcommand{\refeq}[1]{(\ref{#1})}
\newcommand{\mc}[1]{\mathcal{#1}}

\title{Max-Flow Protection using Network Coding}

\author{Osameh M. Al-Kofahi\hspace{1 in}Ahmed E. Kamal\\Department of Electrical and Computer Engineering, Iowa State University, Ames, IA 50010}

\begin{document}
\maketitle

\begin{abstract}

In any communication network, the maximum number of link-disjoint paths between any pair of communicating nodes, S and T, is limited by the S-T minimum link-cut. Multipath routing protocols have been proposed in the literature to make use of these S-T paths in enhancing the survivability of the S-T information flow. This is usually accomplished by using a subset of these paths to forward redundant data units or combinations (if network coding is allowed) from S to T. Therefore, this enhancement in survivability reduces the useful S-T information rate. In this paper we present a new way to enhance the survivability of the S-T information flow without compromising the maximum achievable S-T information rate. To do this, bottleneck links (in the min-cut) should only forward useful information, and not redundant data units. We introduce the idea of extra source or destination connectivity with respect to a certain S-T max-flow, and then we study two problems: namely, pre-cut protection and post-cut protection. Although our objective in both problems is the same, where we aim to maximize the number of protected paths, our analysis shows that the nature of these two problems are very different, and that the pre-cut protection problem is much harder. Specifically, we prove the hardness of the pre-cut protection problem, formulate it as an integer linear program, and propose a heuristic approach to solve it. Simulations show that the performance of the heuristic is acceptable even on relatively large networks. In the post-cut problem we show that all the data units, forwarded by the min-cut edges not incident to T, can be post-cut-protected.
\end{abstract}

\section{Introduction}
\label{intro}

The survivability of an information flow between two terminal nodes, S and T, can be enhanced by using part of the available network resources (bandwidth) to forward redundant information from S to T. Depending on the used survivability mechanism, the redundant information can be used to recover from data corruption if, for example, a Forward Error Correcting code (FEC) is used, or it can be used to recover from network component failures, if a proactive protection mechanism is used. In proactive protection, traditionally $k$ edge-disjoint S-T paths are used to forward $k$ copies of the same data unit from S to T, which guarantees the successful delivery of data if at most $k-1$ link failures occurred in the network. This is usually accomplished by means of a multipath routing protocol, such as MDVA \cite{SJ01} in wired networks or AOMDV \cite{MS01} in ad hoc wireless networks. The maximum number of edge-disjoint S-T paths is limited by the minimum S-T link-cut, which is defined as the smallest set of links that, when removed, all the S-T paths become disconnected. Let $h$ denote the value of the S-T min-cut. Then, if we want to forward data units from S to T and protect them against $q$ failures, we cannot send more than $k = \lfloor\frac{h}{q+1}\rfloor$ data units since $q+1$ copies of each data unit should be forwarded. 

It is clear that traditional proactive protection approaches are very demanding and waste a lot of resources. Even if $q=1$, at least $50\%$ of the used network resources will be wasted to deliver the redundant information, which reduces the useful S-T information rate by at least $50\%$. Network coding \cite{RR00} can be used to overcome this problem in traditional proactive protection schemes. The basic idea of network coding is that it allows intermediate network nodes to generate combinations from the original data units, instead of just forwarding them as is. Therefore, to recover $k$ data units at the destination node T, $k$ linearly independent combinations in the $k$ data units should be delivered to T. That is, if we want to forward data units from S to T and protect them against $q$ failures, we can send at most $k = h - q$ data units. Note that this is done by designing a network code that creates $k+q$ combinations at intermediate network nodes such that any $k$ of them are solvable, which means that it is enough to receive only $k$ combinations to recover the $k$ data units at T. This simple analysis shows that the useful information rate of network coding-based protection is better than that of traditional protection approaches as long as $h > q + 1$, which is usually the case. Examples of network coding-based protection can be found in \cite{A09,OA09,OA08,AC07}.

Network coding-based protection and traditional protection schemes, provide end-to-end protection of the whole S-T paths used to forward useful data from S to T. In these approaches, the more we enhance the S-T flow survivability, the more we reduce the useful S-T information rate. This is because such approaches treat all network links equally, i.e., bottleneck links (that belong to the min-cut) as well as non-bottlenecks are used to forward redundant data units or combinations. Usually, most of the links in a network are not bottleneck links, which means that link failures are more likely to affect non-bottleneck links than links in the min-cut. Therefore, we can enhance the survivability of the S-T information flow without reducing the useful S-T rate below the max-flow, if we provide protection to the non-bottleneck links only. We call this kind of protection \emph{Max-flow protection} because the max-flow can still be achieved under these conditions as long as no link in the min-cut fails. Note that max-flow protection can be transparently combined with end-to-end protection if needed. In this paper, we focus our analysis on the problem of max-flow protection only, and we do not consider combining it with traditional protection schemes. To the best of our knowledge the problem of max-flow protection has not been studied before. 

The rest of this paper is organized as follows. Section \ref{Sec:prelim} presents the terminology and definitions that will be used throughout the paper. The problems of pre-cut and post-cut protection are presented in Section \ref{Sec:problem}. In Section \ref{Sec:preCut} we study the pre-cut protection problem and prove its hardness. The problem of pre-cut protection is formulated as an Integer Linear Program (ILP) in Section \ref{Sec:ILP}. A 3-phase heuristic approach to solve the pre-cut protection problem is described in Section \ref{Sec:Hrstc}. Section \ref{Sec:postCut} discusses the post-cut protection problem. Finally, Section \ref{Sec:conc} concludes the paper.
%

\section{Preliminaries}
\label{Sec:prelim}

We represent a network by a directed acyclic graph G=(V,E), where V is the set of network nodes and E is the set of available links, where each link is assumed to have unit capacity. The network has a source node (S) that wants to send data to a destination (T), where the S-T max-flow is assumed to be $h$. We assume that a multipath routing protocol is used, e.g., \cite{SJ01} or \cite{MS01}, and the source is fully utilizing the available connectivity by sending $h$ data units to the destination simultaneously. To simplify the analysis, we assume that the network has a single cut. In the rest of this section we define the meaning of extra connectivity with respect to the S-T max-flow. After that we discuss some of the properties of nodes with extra connectivity.

\subsection{Terminology}
Let $f^{(A)}(B)$ denote the max-flow from the nodes in set A to the nodes in set B on a directed graph, which can be calculated by computing the max-flow between a virtual source/sink pair, such that the virtual source is connected to the nodes in A with infinite capacity edges and the virtual sink is connected to the nodes in B with infinite capacity edges also. Let $h = f^S(T) $ be the S-T max-flow. We define the following:

\begin{enumerate}
\item A node with Extra Source Connectivity (wESC) is a node, $u$, that satisfies the following conditions:
\begin{itemize}
\item $f^{S}(u,T) > h$, and $f^{(S,u)}(T) = h$.
\end{itemize}
\item A node with Extra Destination Connectivity (wEDC) is a node, $v$, that satisfies the following conditions:
\begin{itemize}
\item $f^{S}(v,T) = h$, and $f^{(S,v)}(T) > h$.
\end{itemize}
\item A node with No Extra Connectivity (wNEC) has:
\begin{itemize}
\item $f^{S}(v,T) = f^{(S,v)}(T) = h$.
\end{itemize}
\end{enumerate}

Of course, a node with both extra source and extra destination connectivity cannot exist, because this contradicts the assumption that the max-flow equals h. Consider the graph G in Figure \ref{origG}. The S-T max-flow in G is 4, which implies that four data units can be forwarded from S to T on four link-disjoint paths. Assume we found the following paths, $P_1 = \{S \rightarrow A \rightarrow E \rightarrow J \rightarrow T\}$ that forwards data unit $w$, $P_2 = \{S \rightarrow B \rightarrow F \rightarrow G \rightarrow T\}$ that forwards data unit $x$, $P_3 = \{S \rightarrow F \rightarrow H \rightarrow T\}$ that forwards data unit $y$, and  $P_4 = \{S \rightarrow D \rightarrow I \rightarrow T\}$ that forwards data unit $z$. Each path $P_i$ contains a cutting edge $C_i$, which , if deleted, will result in reducing the max-flow by exactly 1 unit of flow because path $P_i$ will be disconnected and cannot be reestablished in any way. In our example, $P_1$ contains $C_1=\{(J,T)\}$, $P_2$ contains $C_2=\{(G,T)\}$, $P_3$ contains $C_3=\{(F,H)\}$, and $P_4$ contains $C_4=\{(I,T)\}$. Note that the min-cut may not always be unique, but in this paper we assume that the graph under consideration has only one cut.

%
%

\subsection{Properties of nodes wESC/wEDC}
\label{SubSec:properties}
Consider a path that contains a node, u, wESC and a node, v, wEDC. Note that node u must be closer (in number of hops on the path) to the source than v, otherwise the max-flow assumption will be contradicted, as shown in Figure \ref{NodesOnTheSamePath}. 

\begin{figure}[tbh]
\centering
\includegraphics*[scale = 0.25]{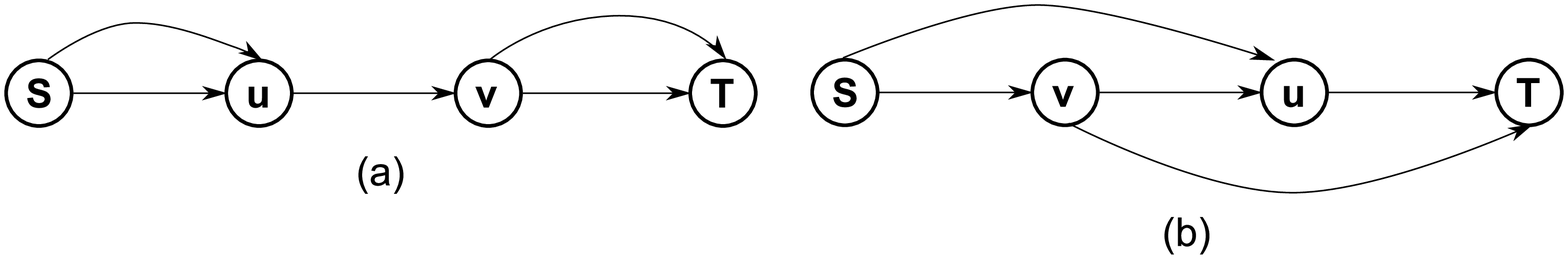}
\caption{In (a) we can say that the S-T max-flow is 1 and that u is a node wESC and v is a node wEDC. However, this is not true in (b), because if  there are two edge-disjoint paths $\{S \rightarrow u \rightarrow T\}$ and $\{S \rightarrow v \rightarrow T\}$, i.e., $u$ is not a node wESC and $v$ is not a node wEDC}
\label{NodesOnTheSamePath}
\end{figure}

In general, removing the min-cut edges (i.e., the edges in $\cup_{i=1}^{h}C_i$) partitions the network into two partitions $A$ and $A'$, such that $S \in A$ and $T \in A'$. Note that, after deleting the min-cut edges, each of the partitions $A$ and $A'$ is a connected component (at least weakly), and that partition $A$ contains nodes wESC, but partition $A'$ contains nodes wEDC. 

\begin{lemma}
Any node $u \in A, u \neq S$ is a node wESC.
\end{lemma}

\begin{proof}
We prove this by contradiction. Let $u \in A, u \neq S$, but $u$ is not a node wESC. Then, $f^S(u,T) = h$, which means that node $u$ cannot receive additional flow from S if the S-T max-flow is established. This implies that either node $u$ is behind the min-cut (i.e., $u \in A'$), which contradicts the starting assumptions, or that there is another min-cut between S and $u$, which contradicts the single min-cut assumption.   
\end{proof}

In a similar fashion, we can prove the following for any node $v \in A', v \neq T$.

\begin{lemma}
Any node $v \in A', v \neq T$ is a node wEDC.
\end{lemma}

In our following discussion we refer to $A$ as the \emph{pre-cut} portion of the network, and to $A'$ as the \emph{post-cut} portion of the network. Figure \ref{genCase} summarizes the previous discussion.

\begin{figure}[tbh]
\centering
\includegraphics*[scale = 0.3]{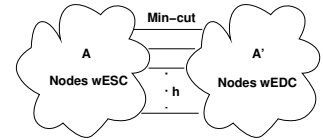}
\caption{Nodes wESC, wNEC and wEDC with respect to min-cuts}
\label{genCase}
\end{figure}


\section{Problem description}
\label{Sec:problem}

The cutting-edges, cannot be protected unless we trade bandwidth for survivability (i.e., unless we use an S-T path to carry redundant information to the destination), which reduces the useful S-T information rate. This tradeoff not only protects the cutting-edges, but also protects any edge carrying data in the network. However, the non-cutting-edges (or a subset of them) can be protected without reducing the S-T information rate, if the graph contains nodes wESC and/or wEDC. For example, nodes E, F, I and J in Figure \ref{origG} are nodes wESC, and node H is a node wEDC. There are four possible ways to utilize the extra source connectivity in Figure \ref{origG}; 1) protect data units $x$ and $y$ by sending $x+y$ to F through C, 2) protect $w$ by sending a duplicate to E through C and F, 3) protect $w$ by sending a duplicate to J through C, F, E and G 4) protect $z$ by sending a duplicate to I through C and F. The first option is better than the other three since sending $x+y$ to F enhances the chances of two data units ($x$ and $y$) to reach T, compared to duplicating $w$ or $z$ alone, which protects a single data unit only. Figure \ref{infoFlowG} shows the first option, and it also shows how to utilize the extra destination connectivity from node H, where H sends a duplicate of $y$ to T through node K.


\begin{figure}[tbh]
\begin{minipage}[b]{0.3\linewidth}
\centering
\includegraphics*[scale = 0.24]{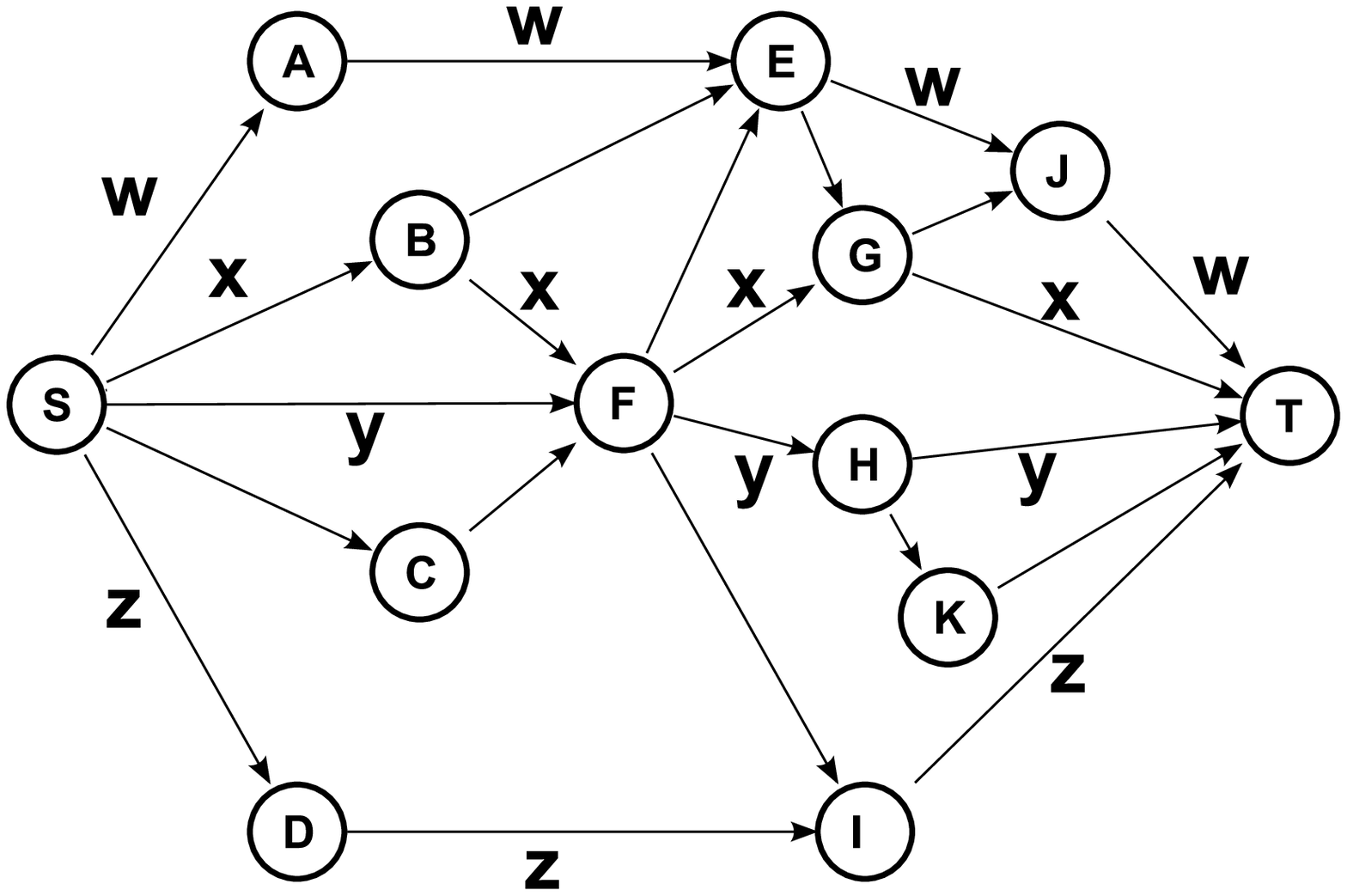}
\caption{Graph G with S-T max-flow = 4}
\label{origG}
\end{minipage}
\hspace{1.5cm}
\begin{minipage}[b]{0.3\linewidth}
\centering
\includegraphics*[scale = 0.24]{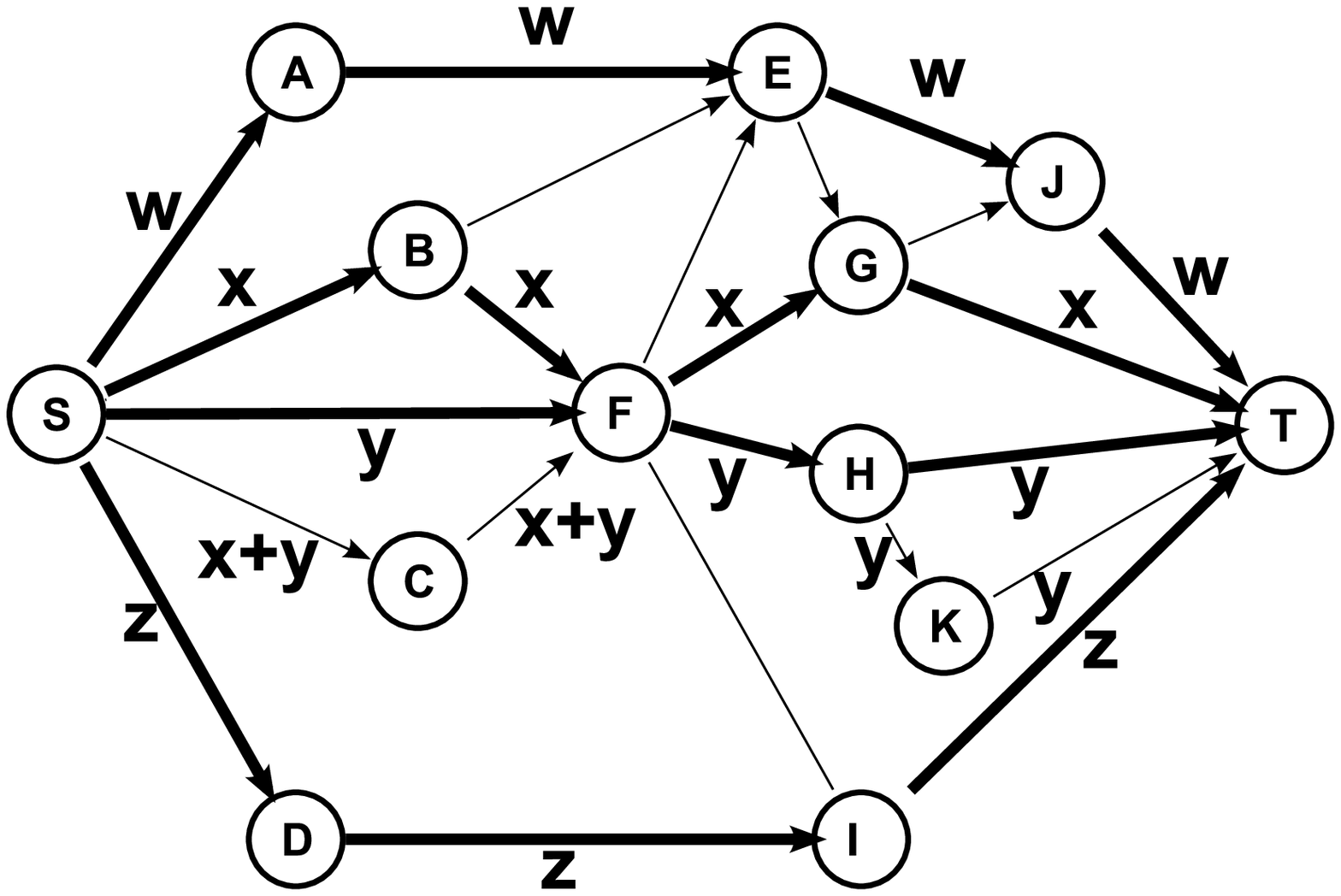}
\caption{Utilizing extra connectivity}
\label{infoFlowG}
\end{minipage}
\end{figure}

In this work, we propose a different way to handle the "survivability vs. bandwidth" trade-off. We propose a new approach to provide protection to the S-T information flow without reducing the useful S-T data rate. Basically, we avoid protecting the bottlenecks in the network (the min-cut links), and we try to efficiently utilize (by using network coding if possible) the available network connectivity before and/or after the bottleneck to provide protection to the non-min-cut links in the graph. We divide the problem into two sub-problems as follows:

\begin{enumerate}
\item Pre-cut protection: Our objective is to maximize the number of pre-cut-protected S-T paths. We show that this problem is NP-hard, and we provide a heuristic to solve it. To evaluate our heuristic we compare its performance to an ILP.
\item Post-cut protection: Similar to the previous objective, we aim to maximize the number of post-cut-protected S-T paths. Let $e_i$ be the closest cutting edge to the destination T on path $P_i$. We show that all the paths that do not have T as the head node of $e_i$ , where $1 \leq i \leq h$, can be post-cut-protected together against at least one failure. 
\end{enumerate}

\section{Pre-Cut: Nodes with Extra Source Connectivity}
\label{Sec:preCut}

As discussed in Section \ref{Sec:prelim}, all nodes wESC are located in the pre-cut portion of the network. Assume that the set $\mc{X}$ contains all the nodes wESC, $\mc{X}=A \backslash S = \{u_1, u_2, \dots, u_{|\mc{X}|}\}$. Then, the following is true:

\begin{equation}
\label{equ:UnrealESC}
(\sum^{|\mc{X}|}_{i=1} f^S(u_i,T)) - |\mc{X}|f^S(T)  \geq f^S(u_1, u_2, \dots, u_{|\mc{X}|}, T) - f^S(T) 
\end{equation}

This is because the extra source connectivity may be shared between the nodes in $\mc{X}$. Therefore, the right hand side of the inequality is what really determines the available extra source connectivity (ESC). This implies that not all nodes wESC in $\mc{X}$ can receive redundant flows from S to be used to protect the S-T max-flow, and thus, a subset $X \subseteq \mc{X}$ should be intelligently selected to receive the available extra source flow and utilize it in the best way possible. Note that the number of nodes in $X$ cannot exceed the extra available connectivity, i.e.:
\begin{equation}
\label{equ:ESC}
ESC = f^S(u_1, u_2, \dots, u_{|\mc{X}|}, T) - f^S(T) \geq |X|
\end{equation}

The selection of $X$ depends on how the S-T max-flow is routed on the graph. Consider the graph in Figures \ref{ex1_a} and \ref{ex1_b}, the S-T max-flow in this network is 2, and there is only one S-T min-cut in the graph, which contains the edges (A,T) and (C,T). Nodes A, B and C are nodes wESC, and the total available extra source connectivity equals $f^S(A,B,C,T) - f^S(T) = 4 - 2 = 2$. Assume that the max-flow is routed as shown in Figure \ref{ex1_a} (the dashed lines), in this case $X_1=\{B,C\}$ since the extra source connectivity is consumed by B and C. Moreover, note that only the path forwarding \textbf{b} can be pre-cut-protected by sending copies of \textbf{b} on $(S,B)$ and $(S,C)$. Now consider the routing shown in Figure \ref{ex1_b}, in this case $X_2=\{A,C\}$. Unlike the previous case, both paths can be pre-cut-protected by sending a second copy of \textbf{a} to A, and a second copy of \textbf{b} to C through B. Obviously, the second routing option is better since it allows the protection of both paths (equivalently both data units), in this sense we say $X_2$ is better than $X_1$. 

\begin{figure}[htp]
\begin{center}
\subfigure[]{
	\label{ex1_a}
	\includegraphics[scale=0.3]{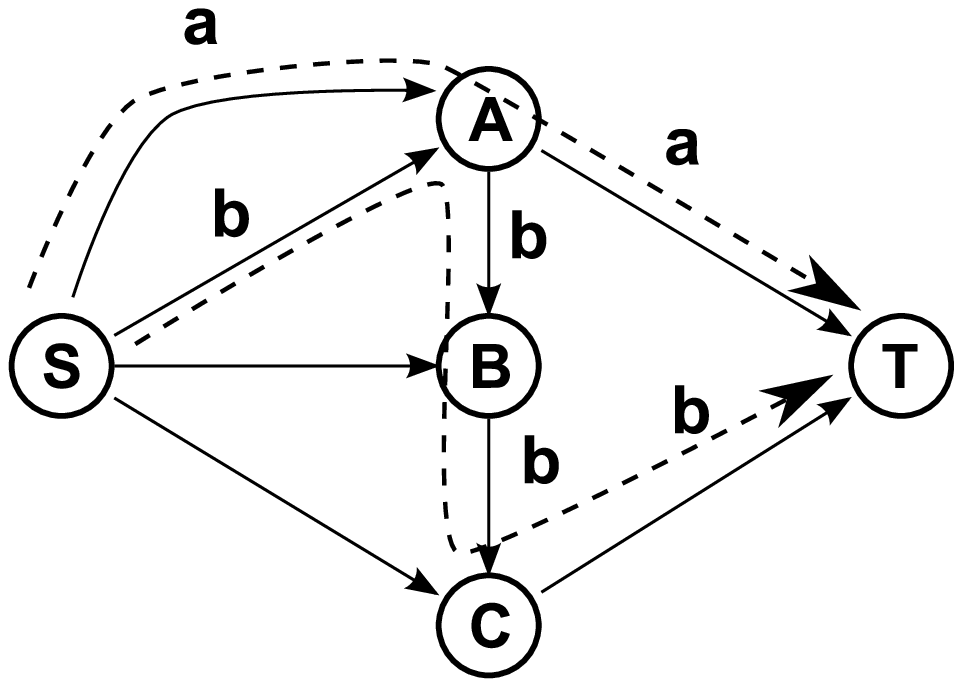}
}
\hspace{1cm}
\subfigure[]{
	\label{ex1_b}
	\includegraphics[scale=0.3]{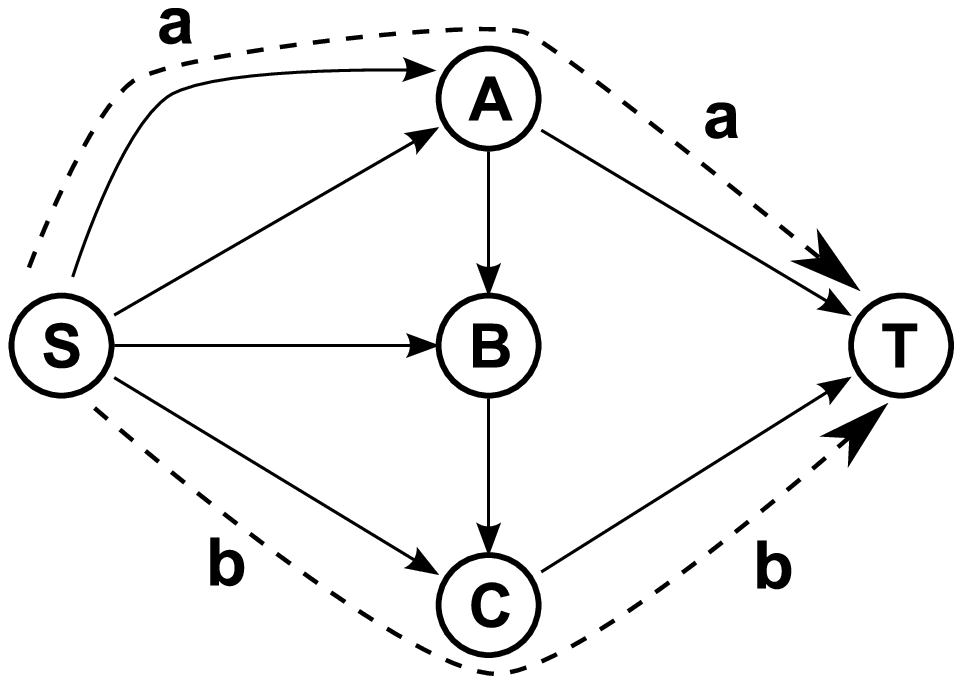}
}
\end{center}
\caption{Routing the max-flow is what determines $X$. In (a) $X=\{B,C\}$, and one path is protected. In (b) $X=\{A,C\}$, and both paths are protected.}
\label{fig:ex1}
\end{figure}

It was shown in the previous example that routing the max-flow and selecting $X$ are inseparable problems, and that routing the S-T max-flow corresponds to selecting $X$. Let us define the extra source connectivity to a node $u$ with respect to the routing of the S-T max-flow in the network as:\[\footnotesize EC(u) = f^S(u,T) - f^S(T)\] We say that an S-T path is pre-cut-protected if a segment of this path in the pre-cut portion of the network is protected. That is, a path is pre-cut-protected if it contains a node wESC with respect to the routing of the S-T max-flow. Therefore, maximizing the number of pre-cut-protected paths means maximizing the number of paths containing nodes wESC. 

For large networks, trying-out all possible routing choices to find the best one that will maximize the number of paths containing nodes wESC is computationally expensive. The following theorem proves that this problem is in fact an NP-hard problem. The full-proof is omitted due to space limitations, and only a sketch of the proof is provided.

\begin{theorem}
Routing the S-T max-flow to maximize the number of S-T paths containing nodes wESC is NP-hard.
\end{theorem}

\begin{proof}
To prove this theorem, we reduce the Maximum Coverage problem with Group budget constraints (MCG) \cite{CA04} to our problem. In the MCG problem, we have a collection of sets $C=\{C_1, C_2, ..., C_m\}$ that are not necessarily disjoint, where each set is a subset of a given ground set $H$. In addition, $C$ is partitioned into disjoint groups $\{G_1,G_2,...,G_n\}$, where each $G_j$ consists of a group of sets in $C$. The problem asks to select $k$ sets from $C$ to maximize the cardinality of their union, such that at most one set from each group is selected. Note that the cover size in the MCG problem is limited by the group budget constraints, and that the number of paths containing nodes wESC is limited by the available extra source connectivity in our problem. To prove the theorem we reduce any instance of the MCG problem to a directed graph with a single cut that translates the group budget constraints into constraints on the available extra source connectivity (similar to the one in Figure \ref{MCG_reduction}). It is now easy to prove that solving the MCG problem solves our problem and vice versa.  

\begin{figure}[tbh]
\centering
\includegraphics*[scale = 0.3]{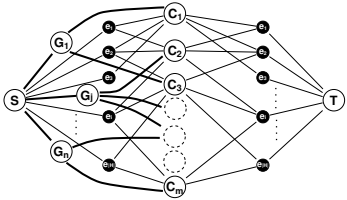}
\caption{Graph resulting from reduction}
\label{MCG_reduction}
\end{figure}

\end{proof}

Note that if network coding was not allowed, then from equation \refeq{equ:ESC} we cannot protect more than ESC data units. Therefore, to utilize the extra source connectivity in a more efficient manner we should apply network coding whenever possible. Network coding can be used if a node wESC, say $u$, lies on more than one S-T path, and has $EC(u) \geq 1$. For example, let $u$ be a node wESC that lies on two S-T paths, and that has $EC(u) = 1$. A network code can be designed to deliver three combinations in two data units to $u$, such that any two combinations are solvable, i.e., two data units are protected from S to $u$ against a single link failure. Note that the number of failures that can be tolerated is at most $EC(u)$. Therefore, the nodes in $X$ should have the following properties:

\begin{enumerate}
\item Each node $u_i \in X$ must have $f^S(u_i) > f^{u_i}(T)$. 
\item The combinations received by a node $u_i \in X$ must be solvable if at most $e = f^S(u_i) - f^{u_i}(T)$ failures occurred on the $f^S(u_i)$ paths from S to $u_i$.
\end{enumerate}
%

The first condition requires the flow from the source to each node $u_i \in X$ to be larger than the flow from that node to the destination. This condition is necessary to introduce redundancy in the forwarding process from S to the nodes in $X$. The second condition can be satisfied by designing a network code that delivers, for each node $u_i$, a set of $f^S(u_i)$ combinations, such that any $f^{u_i}(T)$ combinations of them are solvable. These two conditions allow a node $u_i$ to act as pre-cut decoding node, which can recover the data units sent from S to T through $u_i$, if at most $e = f^S(u_i) - f^{u_i}(T)$ failures occurred on the S-$u_i$ link-disjoint paths, and then send these native data units to T. 

In the next section we present an integer linear program (ILP) formulation of our problem. Solving the ILP will select the routes for the S-T max-flow, and will maximize the number of pre-cut-protected paths (the number of S-T paths containing nodes wESC). 

\section{Integer Linear program Formulation}
\label{Sec:ILP}

We need to maximize the number of S-T paths that contain nodes wESC, regardless of the number of those nodes. We assume that the S-T max-flow equals $h$, and that the flow can take integer values only. Since we are interested in the number of paths containing nodes wESC, we treat each of the $h$ units of flow as a commodity. That is, we have $h$ commodities, each of which is responsible for selecting a single S-T path. The ILP find the routes for these $h$ commodities on a graph with unit-capacity links, such that the number of paths containing nodes wESC is maximized. Let us begin by defining our notation:

\begin{itemize}
\item Let $\sigma_i$ be a binary variable that equals 1 if path $i$ ($P_i$) goes through at least one node wESC, and 0 otherwise. That is, $\sigma_i = 1$ if $P_i$ is pre-cut-protected, and $0$ otherwise.
\item $f^i_{(a,b)}$ is the value of the flow from commodity $i$ on link (a,b). The links forwarding $f^i$ determines $P_i$.
\item $u^i_j$ is the amount of flow $f^i$ entering node j. Although $u^i_j$ is not constrained to be binary, it will be either 1 or 0 since the source sends only one unit of flow $f^i$.
\item $g^j_{(a,b)}$ is the amount of extra flow $g^j$ that is sent from the source to node $j$ on link (a,b). A node that consumes (not forwards) this flow will be included in $X$. 
\item $x_j$ is the amount of flow $g^j$ entering node $j$. Although $x_j$ is not constrained to be binary, it will be either 1 or 0 since the source sends only one unit of flow $g^i$.
\item $\zeta^i_j$ is a binary variable that equals 1 if node $j$ is on $P_i$ and is wESC, i.e., $\zeta^i_j = u^i_jx_j$.
\item $d_j$ is the minimum hop distance of node j from the source, which is a constant that can be computed for each node before solving the ILP, e.g., using Dijkstra's shortest path algorithm. 
\item $\delta^i_j$ is a variable that equals $d_j$ if $\zeta^i_j = 1$, i.e., $\delta^i_j = d_j \zeta^i_j$.
\item $\Omega$ is a very large positive constant. 
\item $w$ is a weighing factor for $\sum \sigma_i$, and is larger than the length of the longest possible path from the source to any node in the network, and can be set to $|E|$. This way the ILP maximizes the length of the protected paths if it does not reduce the number of protected data units.
\end{itemize}

our objective function is:
\begin{equation}
\label{obj}
Maximize ~~~ w\sum^h_{i=1} \sigma_i + \sum_{\forall j}\delta_j
\end{equation}

Subject to,

\begin{equation}
\label{FlowIsH}
{\sum_{\forall(S,b) \in E} f^i_{(S,b)} = 1,~\forall i, ~ where~ 1\leq i \leq h}
\end{equation}

\begin{equation}
\label{FlowIsCnsrvd}
\sum_{\forall(a,b) \in E}f^i_{(a,b)} - \sum_{\forall(b,a) \in E} f^i_{(b,a)} = 0, ~\forall i,~\forall b \in V \backslash \{S,T\}
\end{equation}

\begin{equation}
\label{NodeIsInP}
u^i_j - \sum_{\forall(a,j) \in E}f^i_{(a,j)} = 0 , \forall i, j
\end{equation}

\begin{equation}
\label{gIs1}
\sum_{\forall(S,b) \in E} g^j_{(S,b)} \leq 1,~\forall j
\end{equation}

\begin{equation}
\label{xIsgIn}
x_j - \sum_{\forall(a,j) \in E} g^j_{(a,j)} = 0,~\forall j
\end{equation}

\begin{equation}
\label{gIsCnsrvd}
\sum_{\forall(k,v) \in E}g^j - \sum_{\forall(v,k) \in E} g^j = 0, ~\forall v \in V \backslash \{S,j\}
\end{equation}

\begin{equation}
\label{capCnstrnt}
\sum_{\forall j}g^j_{(a,b)} + \sum_{\forall i} f^i_{(a,b)} \leq 1, ~\forall (a,b) \in E
\end{equation}

\begin{equation}
\label{zetaCnstrnt}
\zeta^i_j - \frac{u^i_j + x_j}{2} \leq 0, ~\forall i, j
\end{equation}

\begin{equation}
\label{sigmaCnstrnt}
\sigma_i - \frac{\sum_{\forall j \in V} \zeta^i_j}{\Omega} < 1
\end{equation}

\begin{equation}
\label{deltaCnstrnt}
\delta^i_j - d_j\zeta^i_j = 0
\end{equation}

Constraint \refeq{FlowIsH} forces the S-T flow to be h, and constraint \refeq{FlowIsCnsrvd} conserves all commodities on all nodes except S and T. \refeq{NodeIsInP} make $u^i_j = 1$ if node $j$ is on path $i$. The extra flow that can be sent to a node wESC is bounded by 1 as shown in constraint \refeq{gIs1}. Constraint \refeq{xIsgIn} sets $x_j$ to 1 if node j receives any extra flow. The extra flow ($g^j$) is conserved at all nodes except the source and node j by constraint \refeq{gIsCnsrvd}. Constraint \refeq{capCnstrnt} guarantees that the link capacity of unit of flow is not exceeded. Constraint \refeq{zetaCnstrnt} sets $\zeta^i_j$ to 1 if node j is on path $P_i$ and is a node wESC. Constraint  \refeq{sigmaCnstrnt} prevents $\sigma_i$ from being 1 if $P_i$ has no node wESC. The value of $\delta^i_j$ is set to $d_j$ if $\zeta^i_j = 1$ by constraint \refeq{deltaCnstrnt}. Note that forcing the extra flow $g^j$ sent to node $j$ to be at most 1 does not affect the ILP optimality, since a path is considered pre-cut-protected if it has a node wESC regardless of the amount of extra flow received at that node.
%
%
%

In the next section we present a heuristic approach to solve the problem of maximizing the number of paths containing nodes wESC. Moreover, we compare the heuristic results to the results from the ILP.

\section{Heuristic approach}
\label{Sec:Hrstc}

Our heuristic works in three phases; the first one greedily selects an initial set $X'$; the second one modifies the flow on the graph (if needed) to guarantee that the S-T max-flow is achieved, and the third one utilizes any remaining connectivity and produces the final set $X$. The first phase works in iterations, where a single node is added to $X'$ in each iteration. Each time we add the node that can send the most flow to the destination, while being able to receive more flow from the source, to satisfy the two conditions stated at the end of Section \ref{Sec:preCut}. If no more nodes satisfy this criteria and the S-T flow is still less than $h$, the second phase is entered. The second phase finds as much augmenting paths as possible from S to T so that the S-T max-flow is maximized. Finally, the third phase checks the nodes in the pre-cut portion of the graph to see if there are any remaining nodes wESC, and makes use of this extra connectivity.

\subsection{Phase 1: Selecting the initial set $X'$}

Recall that if all the min-cut edges are deleted, then the graph will be divided into two partitions $A$ (pre-cut), and $A'$ (post-cut). Note that the routing of the S-T flow in the post-cut portion of the network is independent from the routing of the S-T flow in the pre-cut portion of the network. Therefore, and since the selection of the final set $X$ depends on the routing of the S-T max-flow in the pre-cut portion of the graph, we can simplify the graph under consideration and just focus on the sub-graph, $H$, induced by the nodes in $A$ with a little modification. Specifically, given a directed graph $G(V_G,E_G)$, let $t(u,v)$ represent the tail node of edge (u,v), i.e., node $u$. Also, let $F_S$ be the set of tail nodes on the min-cut edges,i.e., $F_S$ contains the nodes in:\[ \bigcup_{i=1}^h t(C_i)\] where the S-T max-flow $= h$, and $C_i$ is the cutting edge on path $i$ as defined in Section \ref{Sec:prelim}. We transform graph $G$ to $H(V_H,E_H)$ as follows:

\begin{enumerate}
\item Delete the nodes in $\{V_G \backslash A\}$
\item $V_H = \{A, T'\}$, where $T'$ is a dummy destination node.
\item $E_H = \{(u,v)|u,v \in A\} \bigcup \{(u,T')| \forall u \in F_S\}$. Note that $\{S, F_S\} \subset A$.
\end{enumerate}

Each iteration of phase 1 adds the node that can send the most flow to T (or equivalently $T'$), while being able to receive more flow from the source. Let this node be $x$, then it satisfies the following conditions:

\begin{enumerate}
\item $f^x(T') \geq f^u(T'), \forall u \in V_H$
\item $f^S(x) > f^x(T')$
\end{enumerate}

After identifying node $x$, the flow is sent in two steps; in the first step, $(f^x(T') + 1)$ units of flow are sent from S to $x$, and in the second one $f^x(T')$ units of flow are sent from $x$ to $T'$. This way, node $x$ can receive redundant information to protect the $f^x(T)$ path segment from S to $x$. Only one extra unit of flow is sent to $x$ so that the extra source connectivity is fairly divided between the nodes in $X'$ at the end of phase1. 

From a network flows perspective, to forward the flow as described in the previous paragraph, $(f^x(T') + 1)$ units of flow should be sent on $(f^x(T') + 1)$ augmenting paths from S to $x$, and $f^x(T')$ units of flow should be sent on $f^x(T')$ augmenting paths from $x$ to $T'$. Note that since we are working on a residual graph, the paths found from $x$ to $T'$ may contain backward edges, which were used initially to forward flow from S to $x$. If this happens then the flow sent from S to the nodes in $X'$ may be changed and some nodes in $X'$ may not still be nodes wESC. To resolve this issue, we can delete all the edges on the paths found from S to $x$ in each iteration. However, this may reduce our ability to find augmenting paths from S to the nodes in $H$, and thus, may reduce the number of nodes that can be added to $X'$. Therefore, to be able to find augmenting paths without causing any of these problems we work with two copies of $H$. The first one, which we call $H^S$, is used to find paths from S to the nodes in $H$, and the second one, which we refer to as $H^T$, is used to find paths from the nodes in $H$ to $T'$. The links in $H^S$ and $H^T$ are related to each other as follows:

\begin{itemize}
\item After the first step is done, and  $(f^x(T') + 1)$ paths were found from S to $x$ and augmented on $H^S$. Every edge (u, v) in $E_{H^T}$ that corresponds to a backward edge (v, u) in $E_{H^S}$ is deleted.   
\item Similarly, after the second step is completed, and  $f^x(T')$ paths were found from $x$ to $T'$ and augmented on $H^T$. Every edge (u, v) in $E_{H^S}$ that corresponds to a backward edge (v, u) in $E_{H^T}$ is deleted.   
\end{itemize}

In an iteration, if two or more candidate nodes have the same flow to $T'$, the tie is broken in favor of the largest minimum hop distance from the source, i.e., the one with the largest $d^S(u)$ is chosen to be added to $X'$. After that, if two or more nodes have the same flow and minimum hop distance a node is chosen randomly. Taking this into consideration, phase 1 ends when no more nodes can be added to $X'$. 


\subsection{Phase 2: Maximizing the S-T flow}

The resulting S-T flow from phase 1 equals $\sum_{\forall x \in X'} f^x(T')$, which might be less than or equal to $h$ (the max-flow). This is because the extra available connectivity is shared between the nodes in $V_H$. For example, consider the graph in Figure \ref{phase2_ex}, where the S-T max-flow is 2. Phase 1 resulted in adding only one node (F) to $X'$. Assume that node F receives two units of flow from S along the two paths $P_1 = \{S \rightarrow C \rightarrow F\}$ and $P_2 = \{S \rightarrow B \rightarrow A \rightarrow D \rightarrow E \rightarrow F\}$, and sends one unit of flow to T on the direct edge (F, T). The resulting residual graph after augmenting these paths is shown in Figure \ref{phase2_ex_b}, where the backward edges resulting from the augmentation process are shown in boldface. At this point, no more nodes wESC can be added to $X'$ (because the two conditions in the previous subsection are not met for any node), but the S-T flow so far is only equal to 1. Therefore, phase 2 should be entered to maximize the S-T flow. Assume that phase 2 found the path $P_3 = \{S \rightarrow A \rightarrow B \rightarrow E \rightarrow D \rightarrow T\}$, and augmented the flow. After this step, no more S-T paths can be found on the residual graph, which means that the S-T flow is maximized, the resulting residual graph is shown in Figure \ref{phase2_ex_c}. Note that after phase 2, node F still has two link-disjoint paths from S.

\begin{figure}[htp]
\begin{center}
\subfigure[Graph $H$ with S-T max-flow = 2]{
	\label{phase2_ex}
	\includegraphics[scale=0.25]{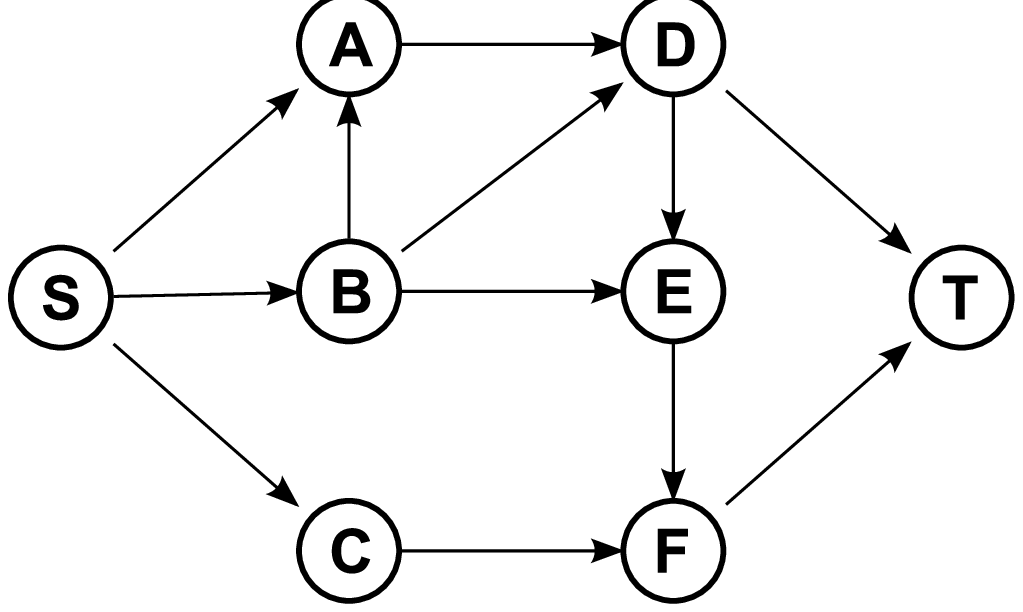}
}
\subfigure[Residual graph after phase1]{
	\label{phase2_ex_b}
	\includegraphics[scale=0.25]{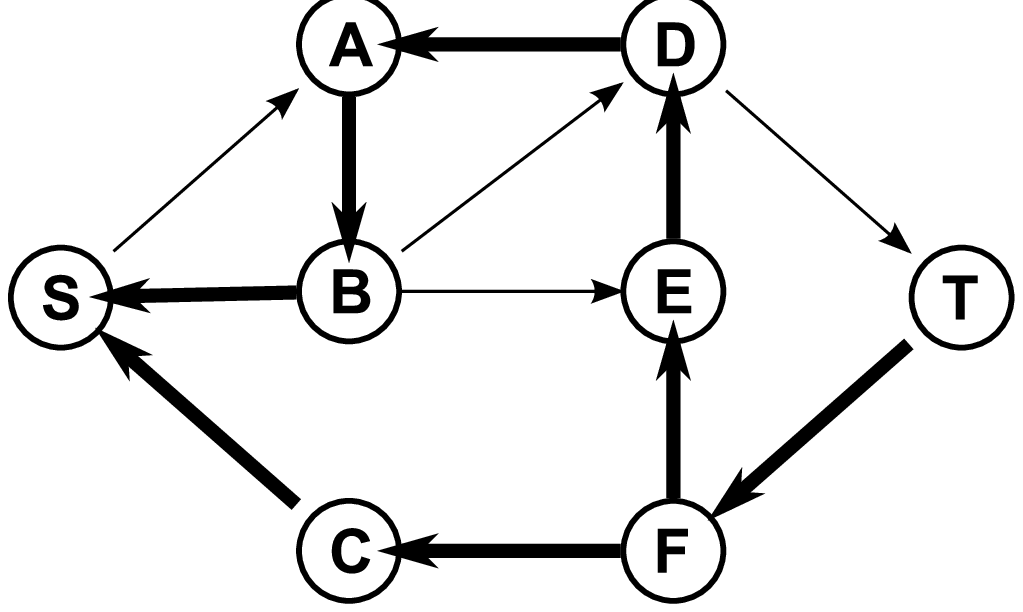}
}
\subfigure[Residual graph after phase2]{
	\label{phase2_ex_c}
	\includegraphics[scale=0.25]{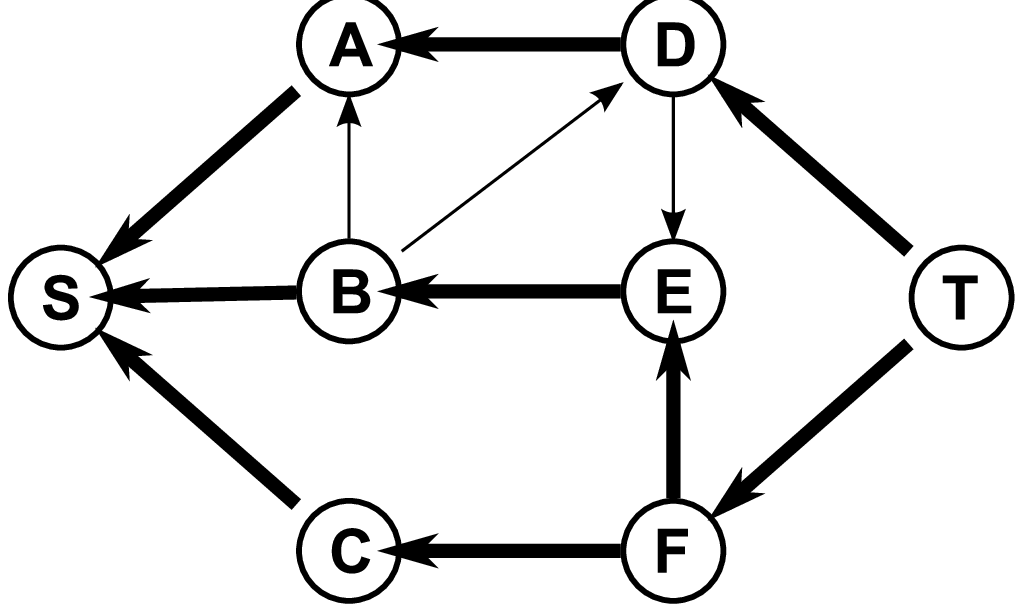}
}
\end{center}
\caption{The operation of the first two phases. (b) node F is added to $X'$, and it receives two units of flow from S and send one unit of flow to T. (c) The S-T max-flow is maximized}
\label{fig:phase2_ex}
\end{figure}

\subsection{Phase 3: Utilizing the remaining ESC}

This phase simply checks if it is possible to send extra flow to any node in $H$ (that lies on at least one path) after the first two phases are finished. If a node $u$ is found to be able to receive extra flow $e'$ from S, then if it is not already in $X'$ it should be added to $X'$. The number of data units node $u$ sends to $T'$ equals $f^u(T')$. The number of data units or combinations it can receive from S is $k = f^u(T') + e'$ if it is not in $X'$, and is $k = f^u(T') + e' + 1$ if it is already in $X'$. If $f^u(T') = 1$ no coding is needed and we need to just send copies of the same forwarded data unit on all the paths to $u$. However, if $f^u(T') > 1$ a network code should be designed to deliver $k$ combinations to $u$ such that any $f^x(T)$ of them are solvable. Algorithm \ref{ALG:SelectingX} summarizes the three phases. 

%
%
%

\begin{algorithm}[!t]
\small
\caption{Selecting set $X$}
\label{ALG:SelectingX}
\begin{algorithmic}[1]
\REQUIRE Graph $H(V_H,E_H)$, $h$ = S-T max-flow
\ENSURE Set $X$ containing nodes wESC 
\STATE $X' = \phi$, $ST\_flow = 0$,  $Phase\_done = 0$
\STATE Create matrices $Flow_S[V_H]$, $Flow_T[V_H]$ \COMMENT{One dimensional matrices initialized to all zeros, to store the final flow from S to each node in $X$, and from each node in $X$ to $T'$. This information will be used for coding later}
\STATE \COMMENT{Phase 1}
\STATE Create graphs $H^S$ and $H^T$, where $V_{H^S} = V_{H^T} = V_{H}$ and $E_{H^S} = E_{H^T} = E_{H}$. 
\WHILE{($Phase\_done==0$)}
\STATE Compute $f^S(u)$ on graph $H^S$, $\forall u \in V_{H^S}$
\STATE Compute $f^u(T')$ on graph $H^T$, $\forall u \in V_{H^T}$
\STATE Select node $x$, where $f^x(T') \geq f^u(T'), \forall u \in V_H$, and $f^S(x) > f^x(T')$
\IF{(No such node exists)}
\STATE $Phase\_done = 1$
\ELSE
\STATE Find $f^x(T')+1$ augmenting paths from S to $x$ on $H^S$
\STATE Delete all forward edges in $H^T$ if they are reversed in $H^S$\COMMENT{due to augmentation}
\STATE Find $f^x(T')$ augmenting paths from $x$ to $T'$
\STATE Delete all forward edges in $H^S$ if the are reversed in $H^T$ 
\STATE $X' = X' \cup \{x\}$
\STATE $ST\_flow = ST\_flow + f^x(T')$
\STATE $Flow_S[x] = f^x(T') + 1$
\STATE $Flow_T[x] = f^x(T')$
\ENDIF
\ENDWHILE
\FORALL{($(u,v) \in E_H$)}
\IF {($(v,u) \in E_{H^S} || (v,u) \in E_{H^T}$)}
\STATE Reverse $(u,v)$ in $H$
\ENDIF
\ENDFOR
\STATE $Phase\_done = 0$ \COMMENT{End of Phase 1, and beginning of Phase 2}
\WHILE {($Phase\_done == 0$)}
\IF{($ST\_flow = h$)}
\STATE $Phase\_done = 1$
\ELSE
\STATE Find an S-$T'$ augmenting path in $H$
\STATE $ST\_flow++$ 
\ENDIF
\ENDWHILE
\STATE $Phase\_done = 0$ \COMMENT{End of Phase 2, and beginning of Phase 3}
\FORALL {($u \in V_H$)}
\STATE Compute $p = f^S(u)$ on the current residual graph of $H$
\IF {($f^S(u) > 0$)}
\STATE Find $p$ augmenting paths from S to $u$ on $H$
\STATE $Flow_S[u] = Flow_S[u] + p$
\ENDIF
\IF {($u \notin X'$)}
\STATE Compute $q = f^{T'}(u)$ on $H^T$
\STATE $Flow_T[u] = Flow_T[u] + q$
\STATE $X' = X' \cup \{u\}$
\ENDIF
\ENDFOR
\RETURN $X'$
\end{algorithmic}
\end{algorithm}

\subsection{Evaluation}

In this section we compare the results from our heuristic to the results from the ILP presented in Section \ref{Sec:ILP}. The heuristic was compared to the ILP in five different cases. Each case represents a different network size, where the number of network nodes $V$ was changed to take the values $\{5, 10, 15, 20, 25\}$. In each case eighty random network instances were generated, and fed to the heuristic and the ILP. Figure \ref{Fig:PerfRatio} shows the ratio between the average number of protected paths by the heuristic and the average number of protected paths by the ILP for the eighty runs. The figure shows that the performance of our heuristic is acceptable, where in the worst case at $V=20$ it was around 77\% of the optimal on average.
 
\begin{figure}[tbh]
\centering
\includegraphics*[scale = 0.35]{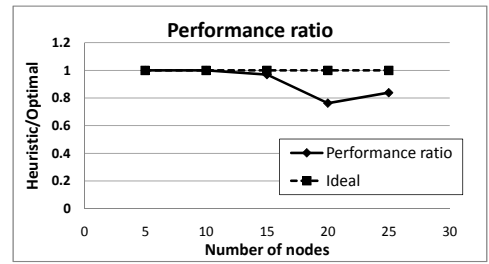}
\caption{Ratio of the number of protected paths by the heuristic to that of the ILP for different number of nodes}
\label{Fig:PerfRatio}
\end{figure}

To gain a better insight on the operation of the heuristic compared to the ILP we measured the S-T max-flow, counted the number of pre-cut-protected paths from the heuristic, and the number of pre-cut-protected paths resulting from the ILP in each time the heuristic and the ILP were executed (on the same network instance). 

\begin{figure*}[htp]
\begin{center}
\subfigure[]{
	\label{SubFig:V_10}
	\includegraphics[scale=0.232]{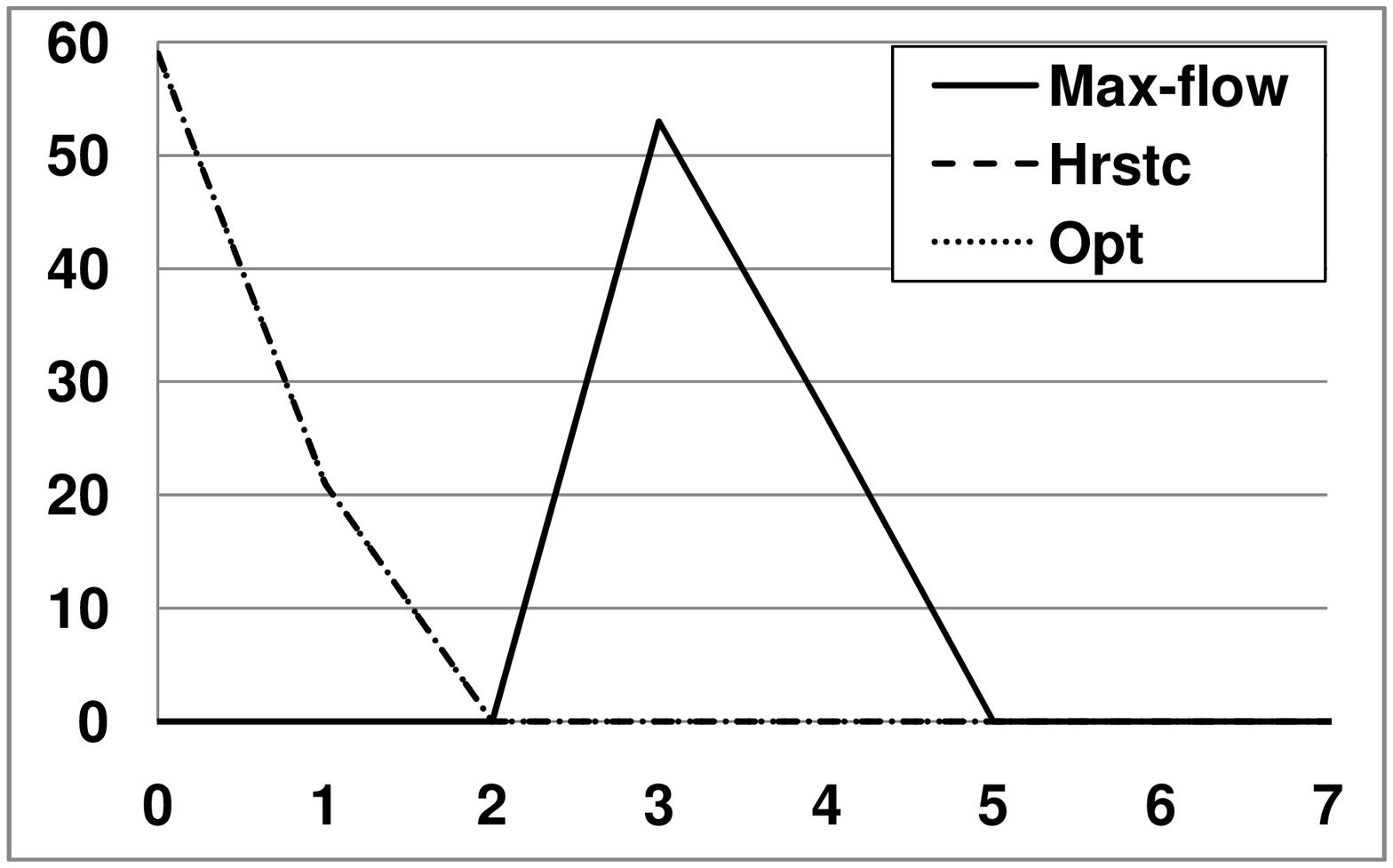}
}
\subfigure[]{
	\label{SubFig:V_15}
	\includegraphics[scale=0.232]{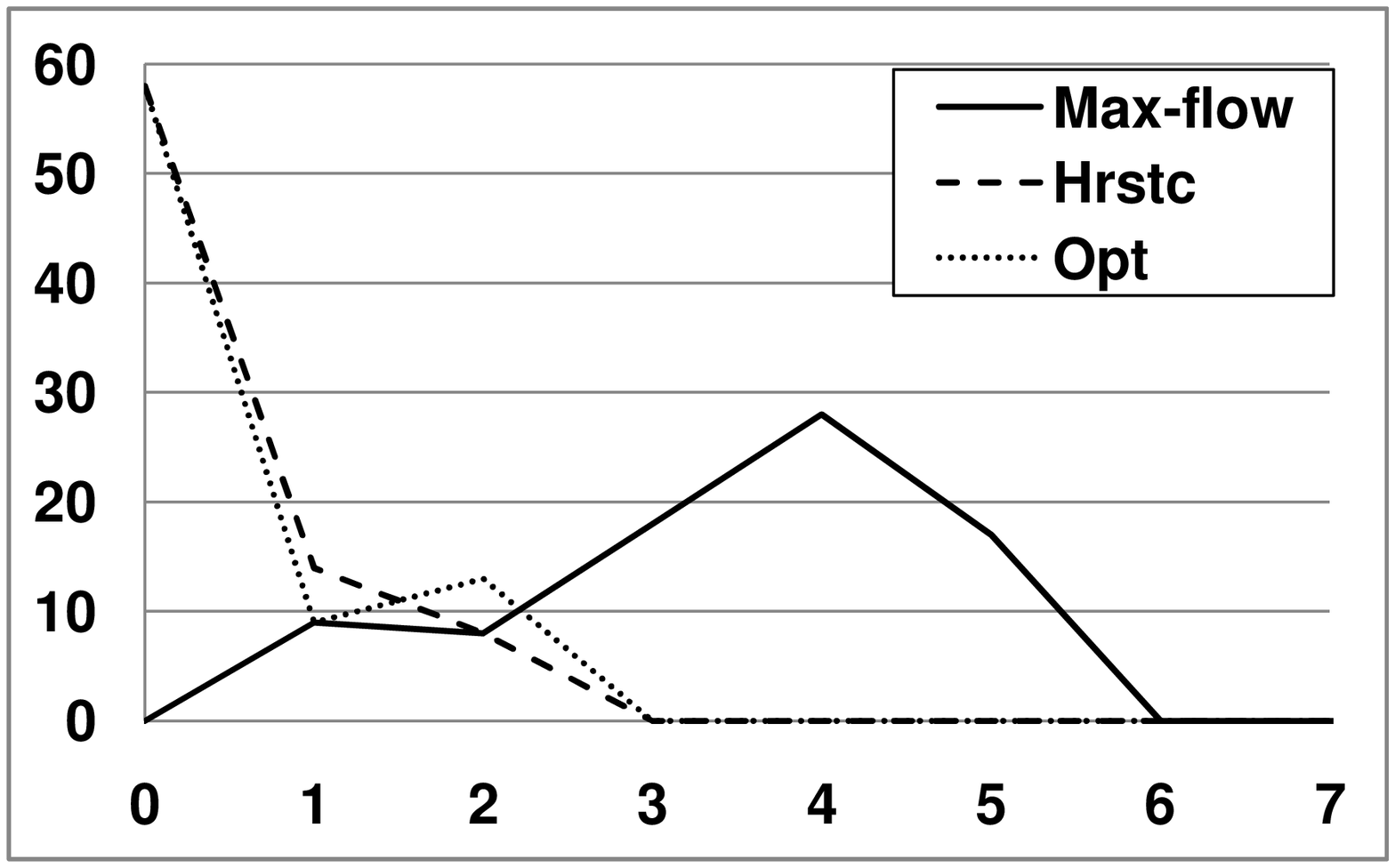}
}
\subfigure[]{
	\label{SubFig:V_20}
	\includegraphics[scale=0.232]{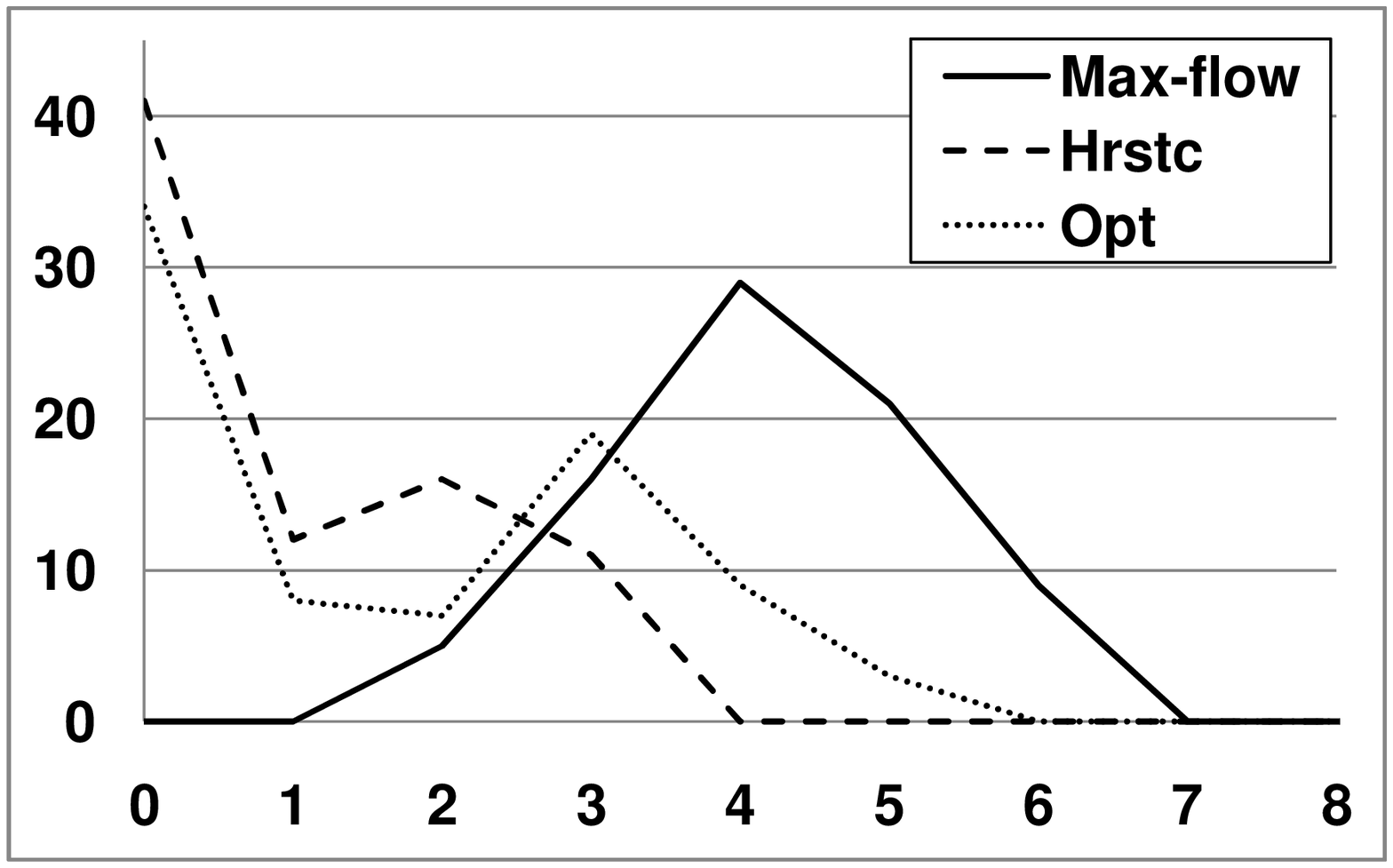}
}
\subfigure[]{
	\label{SubFig:V_25}
	\includegraphics[scale=0.232]{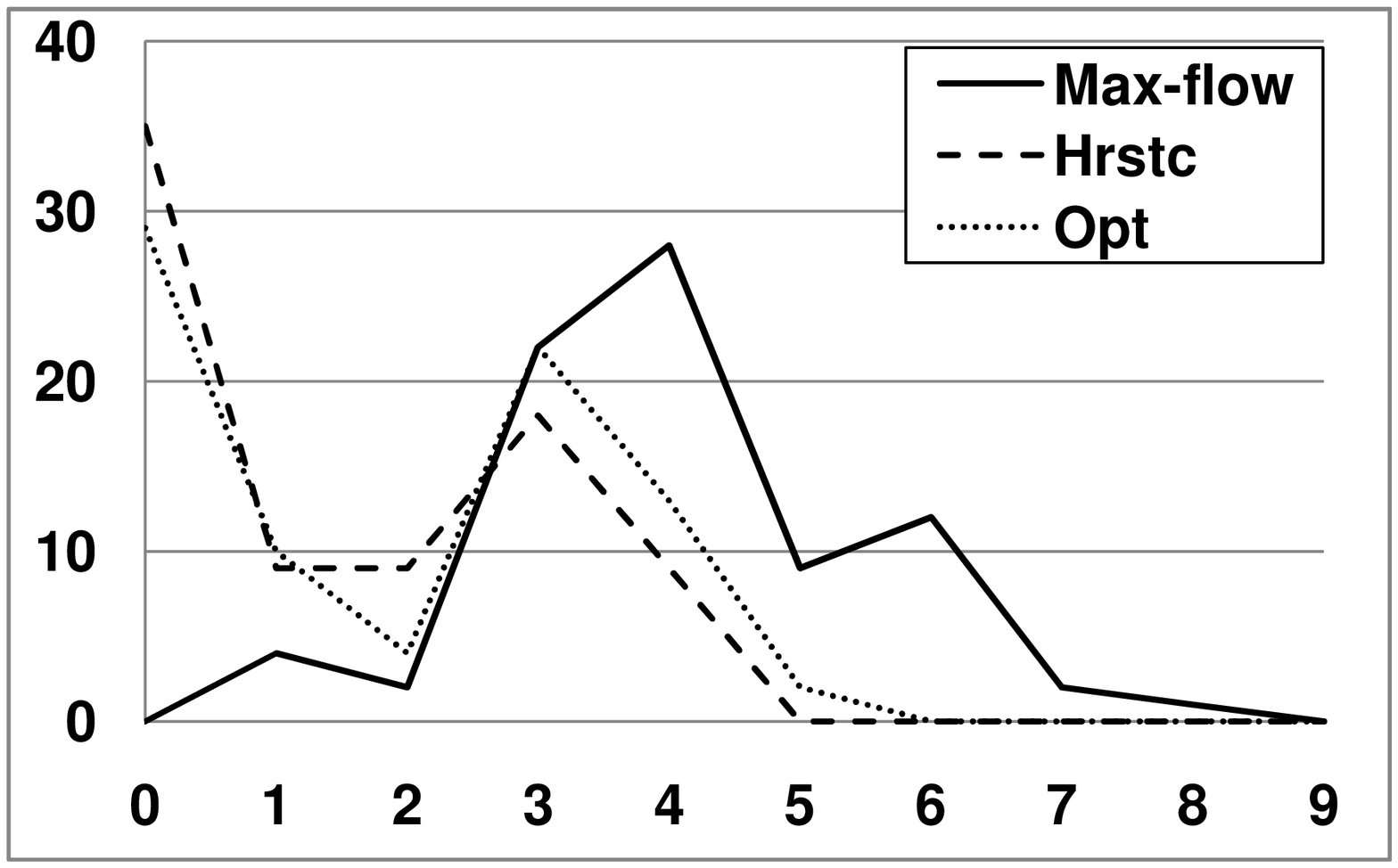}
}
\end{center}
\caption{All figures are histograms, which count three different frequencies: the max-flow, the number of protected paths from the heuristic and the number of protected paths from the ILP. (a) has $V=10$, (b) has $V=15$, (c) has $V=20$ and (d) has $V=25$. The x axis is the number paths either protected or counted in the max-flow, and the y axis is the number of times each number of paths occurred as a max-flow or protected by the ILP or the Heuristic}
\label{Fig:Results}
\end{figure*}
%

The histograms for the cases of V = 10, 15, 20, and 25 are shown in Figures \ref{SubFig:V_10}, \ref{SubFig:V_15}, \ref{SubFig:V_20}, and \ref{SubFig:V_25} respectively. In general, the results from the heuristic are close to those from the ILP. Note that in some cases, the number of times the heuristic is able to protect $X_1$ paths may be larger than the number of times the ILP is able to protect the same number of paths $X_1$. However, this does not invalidate the heuristic because it comes at the price of protecting a larger number of paths $X_2 > X_1$ a fewer number of times. For example, in Figure \ref{SubFig:V_20}, the heuristic was able to protect $X_1 = 2$ paths more than the ILP, but the ILP was able to protect $X_2 = 4$ paths more than the heuristic. 

\subsection{Coding}

The resulting S-$T'$ flow from the heuristic (or the ILP) can be decomposed into two parts; the first, a one-to-many flow from S to the nodes in $X$, and the second is a many-to-one flow from S and the nodes in $X$ to $T'$. The many-to-one flow is not and cannot be coded, since it is composed from the $h$ native data units that are forwarded from S and the nodes in $X$ (possibly after decoding) to $T'$, on $h$ disjoint paths. However, the one-to-many flow from S to the nodes in $X$ can, and should be coded to utilize the extra source connectivity in the most efficient manner. Note that this one-to-many flow is different from normal multicast flow since different data is sent to different nodes. Therefore, a standard multicast network code cannot be used. In fact, the coding in our case is simpler, and needs to be done at a limited number of network nodes as we will show in the following discussion. 

After the heuristic is done and the flow is constructed in the pre-cut portion of the graph. A node $u \in X$ can receive $k+e = Flow_S[u]$ units of flow from S and can send $k = Flow_T[u]$ units of flow to $T'$ (these values were computed in the heuristic). This implies that there are $k+e$ edge-disjoint paths from S to $u$, and $k$ edge disjoint paths from $u$ to $T'$ (or equivalently to T). Note that $k$ represents the number of S-T paths (or data units) going through node $u$, and that $e$ represents the paths used to carry redundant information to $u$.


Let $N_{x_i}$ be the set of 1-hop neighbors of the source on all the $k+e$ paths from S to $x_i$. Assume that all the nodes in $N_{x_i}$ have received the same set of $k$ data units from the source (the $k$ data units on the $k$ S-T paths). To construct a network code that delivers $k+e$ combinations to $x_i$ such that any $k$ of them are solvable using the received data units, we need to assign the proper coding vectors to the nodes in $N_{x_i}$. The coding vectors can be assigned from an $k \times (k+e)$ matrix that has no singular $k \times k$ submatrices, i.e., any $k \times k$ submatrix is invertible. A class of matrices that satisfies this requirement is the Cauchy matrices \cite{FN77}. Therefore, we can simply assign to each node in $N_{x_i}$ a column from a $k \times (k+e)$ Cauchy matrix, such that no two nodes are assigned the same column. 

However, such a coding scheme requires decoding at the nodes in $X$ in each transmission round. An alternative way that will require a fewer number of decoding operations would be to use a systematic code. In a systematic code, $k$ out of the $k+e$ combinations will be trivial combinations, where each of which carries one of the $k$ native data units. In this case, decoding is necessary at a node $x_i \in X$, only if one of the native data units was lost due to a failure on one of the $k$ S-T paths going through node $x_i$. A simple way to do this is presented in \cite{JJ03}. Basically, let $\mc{M}_i$ denote a $k \times (k+e)$ Cauchy matrix with columns representing the coding vectors of the nodes in $N_{x_i}$. We can view $\mc{M}_i$ as two side-by-side matrices $\mc{M}_i = (\mc{M}_{k_i}|\mc{M}_{e_i})$, where $\mc{M}_{k_i}$ is a $k \times k$ matrix , and $\mc{M}_{e_i}$ is a $k \times e$ matrix. Let $\mc{M}'_i$ be the $k \times (k+e)$ matrix resulting from multiplying  $\mc{M}_{k_i}^{-1}$ by $\mc{M}_i$:
\[\mc{M}'_i = \mc{M}_{k_i}^{-1} \times \mc{M}_i = (I_k|\mc{M}_{k_i}^{-1} \times \mc{M}_{e_i}) = (I_k|\mc{M}_{e_i}') = \]
\[
\left[
\begin{array}{cccc|ccc}
1 & 0 & \dots & 0 & \alpha'_{0,k} & \dots & \alpha'_{0,k+e-1} \\
0 & 1 & \dots & 0 & \alpha'_{1,k} & \dots & \alpha'_{1,k+e-1} \\
\vdots & \vdots & \ddots & \vdots & \vdots & \ddots & \vdots \\
0 & 0 & \dots & 1 & \alpha'_{k-1,k} & \dots & \alpha'_{k-1,k+e-1} \\
\end{array} \right]
\]

Since the original matrix $\mc{M}_i$ has no singular submatrices, then the resulting matrix $\mc{M}'_i$ has no singular submatrices also. Note that although the non-singularity property is preserved, the matrix is no longer a Cauchy matrix. Therefore, given that the source has already transmitted the $k$ data units to the nodes in $N_{x_i}$, assigning the columns of $\mc{M}'_i$ to the nodes in $N_{x_i}$ will create $k+e$ combinations such that any $k$ of them are solvable. Moreover, the code is systematic, where out of the $k+e$ combinations there are $k$ trivial combinations, each of which is composed of a single native data unit. 

A special case is when $e=1$. In this case, after the source finishes transmitting the $k$ data units to the nodes in $N_{x_i}$ (where $|N_{x_i}| = k+1$), one of the nodes in $N_{x_i}$ can sum all the received data units and send this sum along with the $k$ native data units on $k+1$ paths to $x_i$.

\section{Post-Cut: Nodes with Extra Destination Connectivity}
\label{Sec:postCut}

Nodes with extra destination connectivity (wEDC) can be found in the post-cut portion of the network only. Nodes wEDC (or a subset of them) can act as post-cut encoding nodes, which create and send redundant combinations to the destination node T to enhance the survivability of the information flow. Note that this case is different from the one considered previously, because all the nodes wEDC are sending their data units to the same destination. Let $hd(u,v)$ denote the head node of edge $(u,v)$, and let $F_T$ be the closest set of nodes wEDC to S, or equivalently the farthest set of node wEDC from T, then $F_T$ contains the nodes in:

\[\bigcup_{i=1}^h hd(C_i)\] where $C_i$ is the cutting edge on path $P_i$. Note that if an edge $C_i$ has $T$ as a head node, then $T \in F_T$, which means that the data unit on the cutting edge $C_i$ is delivered to the destination directly after the cut and cannot be protected. That is, the flow that can be protected from the nodes in $F_T$ is reduced by the number of edges in the cut incident to T. Let $F'_T = F_T \backslash T$, note that $0 \leq |F'_T| \leq h$ (0 when all the nodes in $F_T$ are direct neighbors to T, and $h$ when none of them is a direct neighbor to T), where $h=f^S(T)$. Also, note that since the nodes in $F'_T$ are the head nodes of edges in the min-cut, then we have $f^S(F'_T) = |F'_T|$. That is, each node in $F'_T$ has only one data unit to forward to T, and $|F'_T|$ is the maximum post-cut flow that can be protected. Let $e$ denote the total available extra destination connectivity from the nodes in $F'_T$, then $e$ is calculated as follows \[e = f^{F'_T}(T) - |F'_T|.\] Note that if $F'_T \neq \emptyset$, then $e \geq 1$. If network coding is not allowed, then no more than $e$ data units can be protected. However, if network coding is allowed, we prove that all the data units in $F'_T$ can be protected against at least a single failure:

%
%

\begin{theorem}
Let $F'_T$ be the set of head nodes of the closest min-cut edges to T, such that $T \notin F'_T$. Then if network coding is allowed, the data units at the nodes in $F'_T$ can be protected together against a single failure.  
\end{theorem} 

\begin{proof}
It was shown in \cite{OA09} that a many-to-one flow, similar to the flow from $F'_T$ to T, can be protected against a single link failure (using network coding) if and only if any subset of $k$ source nodes can reach the common destination node through at least $k+1$ edge-disjoint paths. 

Therefore, to prove the theorem we need to prove that any $k$ nodes in $F'_T$ can reach T through at least $k+1$ edge-disjoint paths. We prove this by contradiction. Assume that there is a set, $Q$, of $k$ nodes in $F'_T$ that can reach T through only $k$ edge-disjoint paths, i.e., $f^Q(T) = k$. Then there are $k$ cutting edges on the $k$ paths from the nodes in $Q$ to T, which contradicts the assumption of the single min-cut. Therefore, any $k$ nodes in $F'_T$ must be able to reach the destination node T through at least $k+1$ link disjoint paths, which concludes the proof.
\end{proof}


If $e=1$ we can use the coding tree approach presented in \cite{OA09}. However, if $e > 1$, then to be able to recover the $|F'_T|$ data units if at most $e$ failures occurred in the post-cut portion of the graph, we need two conditions to be satisfied. First, any set of $k$ nodes in $F'_T$ must be able to reach the destination through at least $k+e$ link-disjoint paths. Second, we need to assign coding vectors to the $f^{F'_T}(T)$ combinations such that any $|F'_T|$ vectors from them are linearly independent. Note that if $e > 1$, then the linear independence of any $|F'_T|$ vectors does not necessarily mean that we can recover the $|F'_T|$ data units from any $|F'_T|$ combinations. This is because, when $e$ is larger than 1, the first condition is not necessarily satisfied. To clarify this issue, consider the example in Figure \ref{Fig:EDC_solv}. In the figure, $F'_T = \{A,B,C,D\}$, $f^{F'_T}(T) = 6$ and $e = 2$. The black nodes represent the 6 paths from $F'_T$ to T, and $c_i$ is the combination carried on $P_i$. The links represent the ability of the nodes in $F'_T$ to reach the different paths. If a path $P_i$ can be reached by $k$ nodes in $F'_T$ then $c_i$ is a function of $k$ data units. Note that since $e=2$, the first condition stated above is not satisfied, because nodes C and D can reach T through only three paths not four. To satisfy the second condition, the coding vectors can be chosen as the columns of a $4 \times 6$ Cauchy matrix. Now consider the four combinations $\{c_1,c_2,c_3,c_4\}$. Since $c_1$, $c_2$ and $c_3$ are functions of only two data units A and B (i.e., the coefficients of C and D are zeros), then the three combinations are linearly dependent. However, note that any two combinations of them are linearly independent, because in a Cauchy matrix any square submatrix has full rank (since it is another Cauchy matrix). That is, although the four combinations are in four data units (because of $c_4$), only three are linearly independent and only two are solvable.

\begin{figure}[tbh]
\centering
\includegraphics*[scale = 0.35]{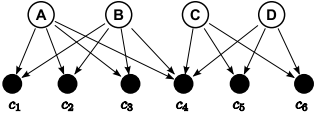}
\caption{$F'_T = \{A,B,C,D\}$, $f^{F'_T}(T) = 6$, and $e = 2$. The combinations $c_1$, $c_2$, and $c_3$ are functions of A and B. The combinations $c_5$ and $c_6$ are functions of C and D. Combination $c_4$ is a function of A, B, C and D. The set $Q = \{c_1,c_2,c_3,c_4\}$ has 3 linearly independent combinations, from which only two can be solved to recover A and B}
\label{Fig:EDC_solv}
\end{figure}
%

If each node $v \in F'_T$ has $f^v(T)$ paths to T that are link disjoint from the paths from all other nodes in $F'_T$ to T, network coding will not be necessary and each node in $F'_T$ can send $f^v(T)$ copies of its data on its $f^v(T)$ paths to T. However, network coding becomes necessary if the paths from the nodes in $F'_T$ to T share links. The first links to be shared are in the link-cut between $F'_T$ and $T$ that is closest to $F'_T$. 
%
%

Let $f^{F'_T}(T) = n$, then there are $n$ edge-disjoint paths $\{P'_1, \dots, P'_n\}$ from $F'_T$ to $T$. Let $C'_i$ denote the cutting edge on path $P_i$ from a node in $F'_T$ to $T$ that is closest to $F'_T$ (if path $P'_i$ has more than one cutting edge). Recall that $C'_i$ is a cutting edge only if the maximum achievable $F'_T$-T flow is reduced by 1. Let $Z$ be the set of coding nodes, which contains the tail nodes of all the $n$ cutting edges as follows:\[\bigcup_{i=1}^nC'_i\]

%

Note that $|Z| \leq n$, and that network coding is not necessary at any of the downstream nodes after $Z$, since the combinations created at the nodes $Z$ will be forwarded to T on $|F'_T|$ edge-disjoint paths. Let $|F'_T| = m$, then a network code can be constructed by assigning each edge $C'_i$, where $1 \leq i \leq n$, a distinct column from an  $m \times n$ Cauchy matrix. The solvability of any $m$ combinations depends on how the nodes in $F'_T$ are connected to $T$ as shown in the previous example. Specifically, let $r$ be the minimum number of solvable combinations in any $m$ combinations, and let $q$ denote the number of failures in the post-cut portion of the graph. Then we are guaranteed the full recovery of the $m$ data units if $q=1$ (by Theorem 2), and we are guaranteed the partial recovery of at least $r$ data units if $q = e$ (by the definition of $r$).

\section{Conclusions}
\label{Sec:conc}
We presented a new protection approach, called max-flow protection, which can enhance the survivability of the whole S-T max-flow. The basic idea is not to protect links in the min-cut, but try to protect all other links if possible. We divided the problem into two problems; pre-cut protection and post-cut protection. Pre-cut protection is NP-hard. Therefore, the problem is formulated as an ILP, and a heuristic is proposed to solve it. We showed that all data units that are not delivered directly to T after the min-cut can be post-cut-protected. Finally, simple network codes are proposed to maximize the number of pre- and post-cut protected paths.

\bibliographystyle{unsrt}
\bibliography{main}

\end{document}